\documentclass[9pt,conference]{IEEEtran}
\usepackage{xcolor}
\usepackage{amsmath,amsfonts}
\usepackage{amsmath}
\usepackage{algorithmic}
\usepackage{algorithm}
\usepackage{array}
\usepackage[font=small]{caption}
\usepackage{textcomp}
\usepackage{stfloats}
\usepackage{url}
\usepackage{verbatim}
\usepackage{graphicx}
\usepackage{cite}
\usepackage{tikz}
\usepackage{amsthm}
\usepackage{subcaption}
\usepackage[normalem]{ulem}
\usepackage{enumitem}
\usepackage{bbm}
\usepackage{flushend}
\usepackage{amssymb}
\newtheorem{remark}{Remark}
\newtheorem{corollary}{Corollary}
\usepackage{mathtools}

\usepackage{breqn}
\usepackage[multiple]{footmisc}
\newcommand{\widesim}[2][1.5]{
  \mathrel{\overset{#2}{\scalebox{#1}[1]{$\sim$}}}
}

\newtheorem{theorem}{Theorem}

\DeclareMathOperator*{\argmax}{arg\,max}

\allowdisplaybreaks

\usepackage{hyperref}
\hypersetup{
	colorlinks   = true, 
	urlcolor     = magenta, 
	linkcolor    = magenta, 
	citecolor   = {red} 
}
\IEEEoverridecommandlockouts
\begin{document}	
	\title{\hspace{-0.29cm}Exploiting Beam-Split in IRS-aided Systems via OFDMA}
	
\author{\IEEEauthorblockN{P. Siddhartha$^\dagger\hspace{0.001cm}^\P$,~L. Yashvanth$^\star\hspace{0.001cm}^\P$,~\IEEEmembership{Student Member,~IEEE}, and~Chandra R. Murthy$^\star$,~\IEEEmembership{Fellow,~IEEE}}
\IEEEauthorblockA{$^\dagger$Dept. of Electrical Engineering, Indian Institute of Technology, Delhi, India. Email: parupudisiddhartha@gmail.com}
\IEEEauthorblockA{$^\star$Dept. of Electrical Communication Engineering, Indian Institute of Science, Bangalore, India.
Email: \{yashvanthl, cmurthy\}@iisc.ac.in}
\thanks{$^{\P}$P. Siddhartha and L. Yashvanth contributed equally to this work.}}
\maketitle
\begin{abstract}
In wideband systems operating at mmWave frequencies, intelligent reflecting surfaces (IRSs) equipped with many passive elements can compensate for channel propagation losses. Then, a phenomenon known as the beam-split (B-SP) occurs in which the phase shifters at the IRS elements fail to beamform at a desired user equipment (UE) over the total allotted bandwidth (BW). Although B-SP is usually seen as an impairment, in this paper, we take an optimistic view and exploit the B-SP effect to enhance the system performance via an orthogonal frequency division multiple access (OFDMA). We argue that due to the B-SP, when an IRS is tuned to beamform at a particular angle on one frequency, it also forms beams in different directions on other frequencies. Then, by opportunistically scheduling different UEs on different subcarriers (SCs), we show that, almost surely, the optimal array gain that scales quadratically in the number of IRS elements can be achieved on all SCs in the system. We derive the achievable throughput of the proposed scheme and deduce that the system also enjoys additional multi-user diversity benefits on top of the optimal beamforming gain over the full BW. Finally, we verify our findings via numerical simulations. 
 \end{abstract}
\section{Introduction}\label{sec:intro}

Intelligent reflecting surfaces (IRSs) are envisioned to enhance the spectral and energy efficiency of wireless systems by reflecting signals in required directions through adjustable phase shifters-enabled passive elements~\cite{IRS_phase_shift,narrowband-2, narrowband-3}. Specifically, at high frequencies, such as mmWaves, IRSs can compensate for larger propagation losses~\cite{Yashuai_PIMRC_2020}. In that case, when a wide bandwidth (BW) is allotted to a user equipment (UE), a phenomenon known as the \emph{beam-split} (B-SP) effect arises, wherein the IRS phase shifters cannot form a beam at the desired UE over the full allotted BW~\cite{VanTress_Array_Proc_2002,beam-squint-1, beam-squint-2, beam-squint-3}. This paper proposes to exploit the B-SP positively and enhance the system throughput via orthogonal frequency division multiple access (OFDMA)

Large IRS apertures are used in high-frequency wideband systems to provide significant beamforming gains at the UE. However, this also makes the channel delay spread monotonically increase with the number of IRS elements and comparable to the system sampling time, known as the spatial wideband effect~\cite{Feifei_Dual_wideband_TSP_2018}.  In such a scenario, the IRS phase shifters, which are primarily used for narrowband beamforming, fail to beamform at a given UE when multiple subcarriers (SCs) spanning a large BW are allotted to that UE, causing the B-SP effect~\cite{VanTress_Array_Proc_2002,beam-squint-1, beam-squint-2, beam-squint-3}. This, in turn, degrades the achievable array gain from the IRS. So, conventional approaches attempt to mitigate the B-SP effects and obtain full beamforming gain on all SCs allotted to any UE. For e.g., in~\cite{Jiang_WCNC_2021_BSQ}, the authors reduce the impact of the B-SP effect in OFDM systems by optimizing IRS phase configurations using appropriate cost functions. In~\cite{Derrick_Wing_Kwan_TCOM_2023,Wanming_TVT_2023,Zhao_WCL_2024,Octavia_CL_2022}, true time-delay (TTD) units are used at each IRS element that equalizes the delay introduced by the channel at that element, in turn, eliminating the B-SP effects. These approaches are computationally expensive or incur hardware and energy costs to precisely control the delays.

In contrast to the works mentioned above, we show that it is possible to positively exploit the naturally occurring B-SP effect and enhance the system throughput. In particular, although the B-SP effect prevents the IRS from forming a beam at a single angle over the full BW, by energy conservation, it forms beams at different angles on across the SCs. Using this B-SP property, we propose to leverage multi-user diversity and improve the system performance through an OFDMA framework. Although~\cite{Beixiong_TWC_2020_OFDMA,Yang_WCL_2020_OFDMA} address some of the problems in IRS-assisted OFDMA, they focus on sub-$6$ GHz bands where B-SP effects are not significant. Considering the mmWave bands, our key contributions are:
\begin{enumerate}[leftmargin=*]
\item Using a max-rate scheduler in OFDMA, and when the IRS configurations are randomly sampled from an appropriate distribution, we show that, on every SC, the probability that at least one UE will achieve nearly full beamforming gain from the IRS increases as the number of UEs increases (See Theorem~\ref{thm_prob_success}.)
\item We derive the achievable system throughput and show that a peak data rate that scales log-quadratically in the total number of IRS elements can be obtained on all SCs in addition to multi-user diversity benefits, with a large number of UEs. (See Theorem~\ref{thm:Rate_scaling}.)
\end{enumerate}
We empirically illustrate our findings and show that the proposed approach exploits B-SP effectively and provides full beamforming gain on all SCs, in turn enhancing the overall spectral efficiency.

\section{System Model and Problem Statement}\label{sec:sys_model}

We consider a downlink wideband mmWave OFDM system with $K$ UEs and $N$ SCs, which span a BW equal to $W$. The system is assisted by an $M$-element IRS implemented as a uniform linear array (ULA) of inter-element spacing $d$. The channel impulse response in the baseband domain from the BS to $m$th IRS element is given by~\cite{beam-squint-1}
\begin{equation}\label{eqn:1}
    g_m(t) =  \bar{\alpha} e^{-j2\pi f_c \tau_{m}^{\mathsf{Tx}}} \delta(t-\tau_{m}^{\mathsf{Tx}}),
\end{equation}
where $f_c$ is the carrier frequency, $\bar{\alpha}$ represents the path loss, and  $\tau_{m}^{\mathsf{Tx}}$ is the delay in the channel from BS to $m$th IRS element. 
Similarly, the channel impulse response from the $m$th IRS element to UE-$k$ is
\begin{equation}\label{eqn:2}
    u_{m,k}(t) = \bar{\beta}_{k} e^{-j2\pi f_c \tau_{m,k}^{\mathsf{Rx}}} \delta(t-\tau_{m,k}^{\mathsf{Rx}}),
\end{equation}
where $\bar{\beta}_{k}$ represents the path loss, and $\tau_{m,k}^{\mathsf{Rx}}$ is the delay in the channel from $m$th IRS element to UE-$k$.
 So, the received signal at UE-$k$ is\footnote{We assume only one propagation path in the BS-IRS and the IRS-$\text{UE-}k$ channels and neglect the direct path between the BS and the UEs due to the high attenuation losses faced by the mmWave signals \cite{pathloss},\cite[Sec.~IV.A]{Yashvanth_TCOM_2024}.}
\begin{equation}
    y_k(t) = \sum_{m=1}^M \left(e^{j\phi_m} (u_{m,k}(t)*g_m(t))\right) * s(t) + n_k(t),
\end{equation} 
where $s(t)$ is the transmitted signal from the BS, $\phi_m$ is the phase shift introduced at the $m$th IRS element, and $n_k(t)$ is the additive noise at UE-$k$. Thus, the end-to-end channel from BS to UE-$k$ is
\begin{equation}\label{eqn:5}
    h_{k}(t) =  \sum\nolimits_{m=1}^{M} e^{j\phi_m} 
     \bar{\alpha}\bar {\beta}_{k} e^{-j2\pi f_c (\tau_{m}^{\mathsf{Tx}} + \tau_{m,k}^{\mathsf{Rx}})} \delta(t- \tau_{m}^{\mathsf{Tx}} - \tau_{m,k}^{\mathsf{Rx}}). 
\end{equation}
For simplicity, we define $\tau^{\mathsf{Tx}} \triangleq \tau_{1}^{\mathsf{Tx}}$ and $\tau_{k}^{\mathsf{Rx}} \triangleq \tau_{1,k}^{\mathsf{Rx}}$. Then, we can show that the channel delays in~\eqref{eqn:1} and~\eqref{eqn:2} can be decomposed as~\cite{beam-squint-1}
\begin{align}
    \tau_{m}^{\mathsf{Tx}}& = \tau^{\mathsf{Tx}} + (m-1)\frac{d\sin(\chi)}{c} = \tau^{\mathsf{Tx}} + (m-1)\frac{\varphi^{\mathsf{Tx}}}{f_c}, \label{eqn:6}\\
    \tau_{m,k}^{\mathsf{Rx}} &= \tau_{k}^{\mathsf{Rx}} - (m-1)\frac{d\sin(\vartheta_{k})}{c} = \tau_{k}^{\mathsf{Rx}} - (m-1)\frac{\varphi_{k}^{\mathsf{Rx}}}{f_c}, \label{eqn:7}
\end{align}
where $\chi,\vartheta_k$ denote the angle of arrival (AoA) from BS to IRS and the angle of departure (AoD) from IRS to UE-$k$, respectively, and $c$ is the speed of light. Further, $\varphi^{\mathsf{Tx}} \triangleq \frac{d\sin(\chi)}{\lambda_c}$, $\varphi_{k}^{\mathsf{Rx}} \triangleq \frac{d\sin(\vartheta_{k})}{\lambda_c}$ represent the normalized AoA and AoD, respectively, where $\lambda_c$ is the carrier wavelength.  
Substituting \eqref{eqn:6} and \eqref{eqn:7} in \eqref{eqn:5}, we get
\begin{equation}\label{eqn:8}
    h_{k}(t) = \sum\nolimits_{m=1}^{M} e^{j\phi_m} 
     \alpha \beta_{k} e^{-j2\pi (m-1) (\varphi^{\mathsf{Tx}} - \varphi_{k}^{\mathsf{Rx}})}   \delta(t- \tau_{m}^{\mathsf{Tx}} - \tau_{m,k}^{\mathsf{Rx}}),
\end{equation}
where $\alpha \triangleq \overline{\alpha} e^{-j2\pi f_c \tau^{\mathsf{Tx}}}$, and
$\beta_{k} \triangleq \overline{\beta}_{k}e^{-j2\pi f_c \tau_{k}^{\mathsf{Rx}}}$ are the complex channel gains. Further, we model $\alpha \sim \mathcal{CN}(0,\rho_1G_{\mathsf{Tx}})$ and $\beta_k \sim \mathcal{CN}(0,\rho_{2,k}G_{\mathsf{Rx}})$, where $\rho_1, \rho_{2,k}$ represent the link path losses, and $G_{\mathsf{Tx}}, G_{\mathsf{Rx}}$ denote the BS and UE antenna gains, respectively. Now, since the channel in~\eqref{eqn:8} has a finite delay spread, we analyze the frequency representation of the channel. Using the Fourier transform, the channel response at UE-$k$ on frequency $f$ can be written as 
\begin{equation}\label{eqn:9}
    H_{k}(f) = \sum\nolimits_{m=1}^M e^{j \phi_m} \gamma_{k}^{\mathsf{C}} e^{-j2\pi (m-1) \varphi_{k}^{\mathsf{C}}\left(1+\frac{f}{f_c}\right)}e^{-j2\pi f \tau_{
      k}^{\mathsf{C}}},
\end{equation}
where 
    $\tau_{k}^{\mathsf{C}} = \tau^{\mathsf{Tx}} + \tau_{k}^{\mathsf{Rx}}$, 
    $\varphi_{k}^{\mathsf{C}} = \varphi^{\mathsf{Tx}} - \varphi_{k}^{\mathsf{Rx}}$, 
    $\gamma_{k}^{\mathsf{C}} = \alpha \beta_{k}$, represent the delay,  angle, and the complex gain of the cascaded channel from BS to UE-$k$, respectively.
From the network viewpoint, since the cascaded angles at different UEs are randomly distributed, we have $\varphi_k^{\mathsf{C}} \stackrel{d}{=} \tilde{\varphi}_k^{\mathsf{C}} \widesim[2]{\text{i.i.d.}} \mathcal{U}[-1,1]$, where $\stackrel{d}{=}$ stands for  ``equal in distribution"~\cite{Yashvanth_TCOM_2024}. Now, we define the array steering vector of a ULA at frequency $f$ and angle $\varphi_{k}^{\mathsf{C}}$ as $\mathbf{a}\left((1+\frac{f}{f_c}) \varphi_{k}^{\mathsf{C}}\right) \triangleq \left[1,e^{-j2\pi \left(1+\frac{f}{f_c}\right)\varphi_{k}^{\mathsf{C}}}, \ldots,e^{-j2\pi(M-1) \left(1+\frac{f}{f_c}\right)\varphi_{k}^{\mathsf{C}}}\right]^T $, and compactly write the effective channel from BS to UE-$k$ on frequency $f$ as
\begin{equation}\label{eqn:13}
    H_k(f) = \gamma_{k}^{\mathsf{C}} \boldsymbol{\theta}^T\mathbf{a}\left(\left(1+{f}/{f_c}\right) \varphi_{k}^{\mathsf{C}}\right) e^{-j2\pi f \tau_{k}^{\mathsf{C}}},
\end{equation}
where $\boldsymbol{\theta} \triangleq [e^{j\phi_1},\ldots,e^{j\phi_M}]^T$. From~\eqref{eqn:13}, the channel gain at UE-$k$ on frequency $f$ is given by

\begin{equation}\label{eqn_ch_gain}
 |H_k(f)|^2 = \left|\gamma_{k}^{\mathsf{C}}\right|^2 \bigg|\boldsymbol{\theta}^T\mathbf{a}\bigg(\bigg(1+\frac{f}{f_c}\bigg) \varphi_{k}^{\mathsf{C}}\bigg)\bigg|^2.
\end{equation}
Therefore, in OFDM, for the IRS to form a beam at the $n$th SC ($n \in [N]$), the optimal phase of $m$th IRS element is given by\footnote{The phase shifters at the IRS has a frequency flat response; so we tune the IRS to a single frequency. As a result, the IRS cannot perfectly beamform on all SCs unless the channel it optimizes also has a frequency flat response~\cite{Yashvanth_SPCOM_2022}.}
\begin{equation}\label{eqn:16}
    \phi_m = 2\pi(m-1)(\varphi_{k}^{\mathsf{C}})\bigg(1+\frac{f_n}{f_c}\bigg) ,
\end{equation}
where $f_n = \frac{nW}{N} - \frac{W}{2}-\frac{W}{2N}$ is the baseband frequency of $n$th SC. Then using~\eqref{eqn:16} in~\eqref{eqn_ch_gain}, the gain at UE-$k$ on frequency $f$ becomes

 \begin{equation}
     |H^{\mathrm{opt}}_{k}(f)|^2 \approx M^2\left|\gamma_{k}^{\mathsf{C}}\right|^2\sin\!\mathrm{c}^2\left(M(f_n-f)\varphi_{k}^{\mathsf{C}}\Big/f_c\right).
 \end{equation}
Clearly, the channel gain is maximum at $\mathcal{O}(M^2)$ when $f=f_n$ or $\varphi_{k}^{\mathsf{C}}=0$. On SCs with $f\neq f_n$, the gain substantially reduces and even goes to $0$ on some SCs. This is called the \emph{beam-split} effect and leads to degradation of array gain at the UE, as illustrated in Fig.~\ref{fig:Beam_Squint_demo_1}. 
\subsection{Problem Statement}
From the above discussion, allotting the full BW to any single UE (unless $\varphi_{k}^{\mathsf{C}}=0$) gives rise to B-SP and degrades performance. However, in an OFDMA scheme, multiple UEs are served simultaneously by multiplexing them on different SCs~\cite{Beixiong_TWC_2020_OFDMA,Yang_WCL_2020_OFDMA}. This fits naturally into the constraints imposed by the B-SP effect since no single UE is typically allotted the full BW. Hence, we propose to exploit the B-SP effect and achieve a full array gain of $M^2$ on all SCs via OFDMA. 
In this context, our goal is to maximize the system throughput obtained over $T$ time slots with respect to all schedulers $\mathsf{SCH}(n,t)$, where $\mathsf{SCH}: [N] \times [T] \rightarrow [K]$ maps the SC index $n$ and time slot $t$ to a UE index $k$. Mathematically, the optimization problem is
\begin{equation}\label{Prob_statement_Equation}
\hspace{-0.3cm}\max_{\mathsf{SCH}(n,t)} \frac{1}{T}\sum_{t = 1}^T \sum_{n = 1}^N \frac{W}{N}\log_2\left(1 + \frac{P}{N\sigma^2}|H(\mathsf{SCH}(n,t),t,f_n)|^2\right)\!,\!\!\! \tag{P0}
\end{equation}
where $P$ is the total transmit power, $\sigma^2$ is the noise variance per SC, and $H(k,t,f)$ is the channel at UE-$k$ on frequency $f$ at time $t$. 
In the next section, we solve for the optimal scheduler $\mathsf{SCH}(\cdot)$ and subsequently analyze the effectiveness of the scheduler in exploiting the B-SP effect to boost the performance of an OFDMA system.

  \begin{figure}
  \vspace{-0.2cm}
     \centering
     \includegraphics[width=0.9\linewidth]{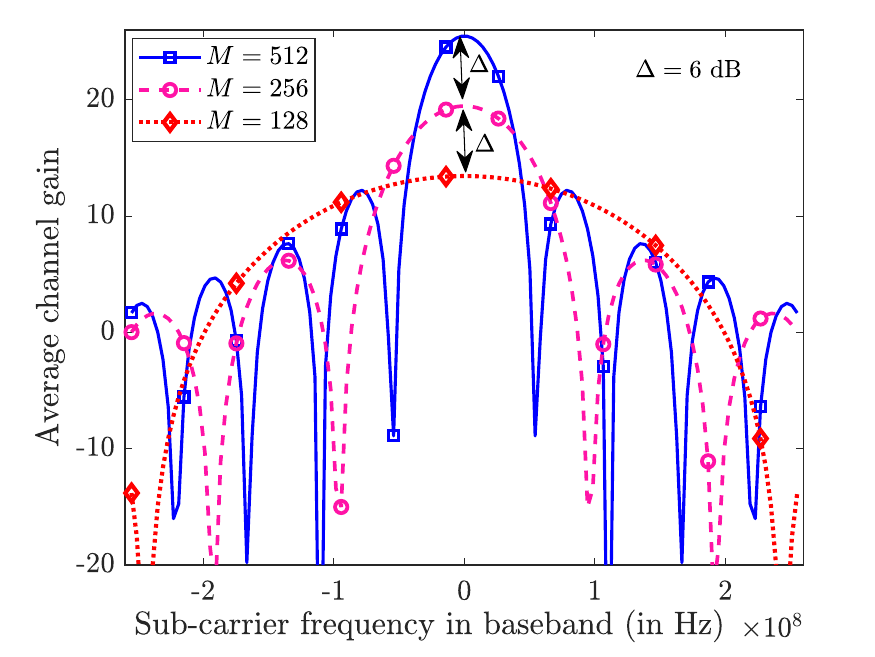}
     \caption{Average channel gain vs. SC frequency when IRS is optimized to $f_n=0$ for different number of IRS elements, $M$ with $W=512$ MHz. Although the gain at $f_n=0$ increases with $M$, the gain on other SCs proportionately degrades as $M$ increases due to the B-SP effect.}
     \label{fig:Beam_Squint_demo_1}
     \vspace{-0.2cm}
\end{figure}

\section{Exploiting Beam Split for Opportunistic OFDMA }\label{sec:problem-statement-soln}
From~\eqref{Prob_statement_Equation}, since the scheduler $\mathsf{SCH}(n,t)$ is decoupled across SCs and time slots, the optimal scheduler as the solution to~\eqref{Prob_statement_Equation} is the one which maximizes the rate/gain on each SC in every slot, i.e., 
\begin{equation}\label{eq_max_rate}
   \forall n \in [N], t\in[T], \  \mathsf{SCH}^{\mathrm{opt}}(n,t) = \argmax_{k \in [K]}\ |H(k,t,f_n)|^2.
\end{equation}
Since multiple UEs are served simultaneously, tuning the IRS configuration to a particular UE leads to sub-optimal performance. Thus, we propose to randomly configure the IRS from a distribution that is aware of the channel model to which the IRS is optimized in every time slot. Specifically, using~\cite[Sec.~III.C]{IRS-OC}, the random IRS configurations are sampled independently across time as
\begin{equation}\label{eqn:22}
    \phi_m(t) = 2\pi (m-1)a(t), \ a(t)\widesim[2]{\text{i.i.d.}} \mathcal{U}[-1,1], \ t\in[T].
\end{equation}
\textcolor{black}{Then, using~\eqref{eqn:22} in~\eqref{eqn:9}, we can show that the channel gain at an arbitrary UE-$q$ on SC-$n$ at time slot $t$ is given by
\begin{equation*}
\left|H(q,t,f_n)\right|^2 = M^2\left|\gamma_{q}^{\mathsf{C}}\right|^2\sin\!\mathrm{c}^2\left(M\left(a(t)-\varphi_q^{\mathsf{C}}\left(1+f_n/f_c\right)\right)\right).
\end{equation*} Thus, for a given $a(t)$ (at time $t$), if there are a large number of UEs, $\forall \ n\in[N]$, there will exist at least one UE-$q_n$ ($q_n \in [K]$) on SC $n$,  such that $a(t)-\varphi_{q_n}^{\mathsf{C}}\left(1+f_n/f_c\right) \approx 0$, and can nearly achieve the full array gain of $M^2$. Then, by scheduling such a UE on every SC, the randomly chosen IRS configuration will be near-optimal to all the UEs scheduled over the full BW~\cite{IRS-OC}.
This way, in the mmWave bands with large BW, every SC witnesses the optimal array gain of $M^2$ (and hence the throughput), in spite of B-SP, by exploiting multi-user diversity via opportunistic OFDMA.}
 
We next formally prove that under the max-rate scheduler, the array gain on every SC, as seen by the network, approaches the maximum value of $M^2$ by leveraging the multi-user diversity and B-SP effects. 
 \begin{theorem}\label{thm_prob_success}
  Let $\mathcal{A}_{k,n}^\epsilon$ denote the event that the array gain on SC-$n$ at UE-$k$ is at least $(1-\epsilon)M^2$ at some time $t$. Then, using a max-rate scheduler with randomized IRS configurations sampled as per~\eqref{eqn:22},
  \begin{equation}\label{eqn_prob_succ}
 \hspace{-0.2cm}P^{\epsilon}_{\textrm{succ}} \triangleq \mathsf{Pr}\left(\bigcap_{n = 1}^{N}\bigcup_{k = 1}^{K} \mathcal{A}_{k,n}^{\epsilon}\right) \geq 1 - N\left(1-\frac{\sqrt{3\epsilon}}{\pi M \left(1+\frac{W}{2f_c}\right)}\right)^{\!\!K}\!\!\!\!.\!
  \end{equation}
 \end{theorem}
 \begin{proof}
 Note that the event $\bigcap_{n = 1}^{N}\bigcup_{k = 1}^{K} \mathcal{A}_{k,n}^{\epsilon}$ collects the set of outcomes such that on every SC $n=1,\ldots, N$, there exists at least one UE among the $K$ UEs such that the array gain at that UE is at least $(1-\epsilon)M^2$. We then make the following observations.
 \begin{align}
 &P^{\epsilon}_{\textrm{succ}} = 1 - \mathsf{Pr}\left(\left(\bigcap_{n = 1}^{N}\bigcup_{k = 1}^{K} \mathcal{A}_{k,n}^{\epsilon}\right)^{\textrm{c}}\right)  \stackrel{(a)}{=} 1 -\mathsf{Pr}\left(\bigcup_{n=1}^N \bigcap_{k=1}^K (\mathcal{A}_{k,n}^{\epsilon})^{\textrm{c}}\right)\nonumber \\
        &\hspace{-0.2cm} \stackrel{(b)}{\geq} 1 - \sum_{n=1}^N \mathsf{Pr}\left(\bigcap_{k=1}^K {(\mathcal{A}_{k,n}^{\epsilon})}^{\textrm{c}}\right) \stackrel{(c)}{=} 1 - \sum_{n=1}^N \prod_{k=1}^{K}\mathsf{Pr}\left({(\mathcal{A}_{k,n}^{\epsilon})}^{\textrm{c}}\right)  \label{eqn_prob_simple_template},
        \end{align}
        where in $(a)$, we used the De-Morgan's law, in $(b)$, we used the union bound, and in $(c)$, we noted that $\left\{\mathcal{A}_{k,n}^{\epsilon}\right\}_{k=1}^{K}$ are independent events. Now, at any given $t$, and hence the IRS configuration, ${(\mathcal{A}_{k,n}^{\epsilon})}^{\textrm{c}}$       
          \begin{align}
        &= \left\{\varphi_{k}^{\mathsf{C}} \in [-1,1]: \sin\!\mathrm{c}^2\left(M\left(a(t)-\left(1+f_n/f_c\right)\varphi_{k}^{\mathsf{C}}\right)\right) \leq 1-\epsilon \right\} \nonumber \\
        &\stackrel{(d)}{=} \left\{\varphi_{k}^{\mathsf{C}} \in [-1,1] : \bigg|a(t)-\left(1+f_n/f_c\right)\varphi_{k}^{\mathsf{C}}\bigg| \geq \sqrt{3\epsilon}/\pi M  \right\}   \\
        &=\left \{\varphi_{k}^{\mathsf{C}} \in [-1,1]: \varphi_{k}^{\mathsf{C}} \notin a(t) \pm \frac{\sqrt{3\epsilon}}{\pi M}\Big/1+f_n/f_c \right\},
    \end{align}
    where in $(d)$, we used Taylor's approximation: $\sin\!\mathrm{c}^2(x)\approx1-\pi^2x^2/3$, which is tight under the regime of interest to us.
    Now, when $\varphi_{k}^{\mathsf{C}} \sim \mathcal{U}[-1,1]$, we can show that $\mathsf{Pr}\left((A_{k,n}^{\epsilon})^{\textrm{c}}\right)=1-\sqrt{3\epsilon}\big/\pi M (1+f_n/f_c)$. Recognizing that $|f_n|\leq W/2, \ n \in [N]$, we further bound $\mathsf{Pr}\left((A_{k,n}^{\epsilon})^{\textrm{c}}\right)$ and plugging the resulting  in~\eqref{eqn_prob_simple_template} yields~\eqref{eqn_prob_succ}.
 \end{proof}
From Theorem~\ref{thm_prob_success}, the probability, $P^{\epsilon}_{\textrm{succ}}$ is non-decreasing in the number of UEs, $K$. In fact, in the limiting scenario, we see that
 \begin{equation}
 \lim_{K\rightarrow \infty} \mathsf{Pr}\left(\bigcap_{n = 1}^{N}\bigcup_{k = 1}^{K} \mathcal{A}_{k,n}^{\epsilon}\right) = 1.
 \end{equation}
Hence, with probability $1$, every SC will witness a $\mathcal{O}(M^2)$ array gain by exploiting the B-SP effect when a max-rate scheduler is employed over many UEs. We now state our next result, which characterizes the achievable system throughput when all UEs have equal path losses.\footnote{We use equal path loss across UEs only to obtain a tractable analytic scaling of the throughput as a function of system parameters, similar to~\cite{Nadeem_TWC_2021}.}
 \begin{theorem}\label{thm:Rate_scaling}
     The system throughput of an OFDMA scheme in a randomly configured IRS (as per~\eqref{eqn:22}) aided wideband system using a max-rate scheduler and equal power allocation across SCs satisfies
    \begin{equation*}\label{eqn:31}
     \!\!\!   \lim_{K \rightarrow \infty} \bigg(R^{(K)} - \mathcal{O}\left\{W \log_2\left(1 + \frac{\rho G_{\mathsf{Tx}}G_{\mathsf{Rx}} P}{N\sigma^2}M^2 \ln{K}\right)\right\}\bigg) = 0.
    \end{equation*}
    where $\rho \triangleq \rho_1\rho_2$ is the common cascaded path loss across the UEs.
\end{theorem}
    \begin{proof} Note that the system throughput with a max-rate scheduler is 
\begin{equation}\label{eqn:32}
    \hspace{-0.1cm}R_{\textrm{MR}} \triangleq \frac{1}{T}\sum_{t = 1}^T \sum_{n = 1}^N \frac{W}{N}\log_2\left(1 + \frac{P}{N\sigma^2}\max_{k \in [K]}|H\big(k, t, f_n)|^2\right).
\end{equation}
Now, considering that the channels are ergodic and using Jensen's approximation over the $\log_2(\cdot)$ function, we simplify \eqref{eqn:32} as
\begin{equation}\label{eq_init_E_val}
\hspace{-0.1cm}\mathbb{E}[R_{\textrm{MR}}] \approx  \sum_{n=1}^N \frac{W}{N} \log_2\left(1 + \frac{P}{N\sigma^2} \mathbb{E}\left[\max_{k \in [K]}|H(k,t, f_n)|^2\right]\right).
\end{equation}
To characterize the expectation in~\eqref{eq_init_E_val}, we can show $ |H(k,t,f_n)|^2$
\begin{align}
             &   \approx M^2 \rho_1\rho_2G_{\mathsf{Tx}}G_{\mathsf{Rx}}\left|\tilde{\gamma}^{\mathsf{C}}_{k}\right|^2\sin\!\mathrm{c}^2(M (a(t)-\varphi^{\mathsf{C}}_{k}(1+f_n/f_c)),\label{eq_channel_gain_with_sinc}
 \end{align}
 where we define $\gamma^{\mathsf{C}}_{k} = \sqrt{\rho_1\rho_2G_{\mathsf{Tx}}G_{\mathsf{Rx}}}\tilde{\gamma}^{\mathsf{C}}_{k}$ with $\rho_2 = \rho_{2,k} \forall k$ under the equal path loss assumption. Here, $\tilde{\gamma}^{\mathsf{C}}_{k} \widesim[2]{\text{i.i.d.}} \mathcal{P}(0,1)$, where $\mathcal{P}(\mu,\sigma^2)$ represents the distribution of product of two complex normal random variables with mean $\mu$ and variance $\sigma^2$. 
 From Theorem~\ref{thm_prob_success}, as $K \rightarrow \infty$, we know that at least one UE will be near the beamforming configuration on all SCs. Then, the maximum of terms given in~\eqref{eq_channel_gain_with_sinc} is the maximum over the channel gains among those UEs for whom the IRS phases are in near-beamforming configurations. So, we have
 \begin{equation}\label{eq_expected_OS}
 \mathbb{E}\left[\max_{k \in [K]}|H(k,t, f_n)|^2\right] \approx M^2 \rho_1\rho_2G_{\mathsf{Tx}}G_{\mathsf{Rx}}\mathbb{E}\left[\max_{k\in [k]} \left|\tilde{\gamma}^{\mathsf{C}}_{k}\right|^2\right] + \mathcal{O}(1).
 \end{equation}
 Now, to characterize the expected value of the order statistic in~\eqref{eq_expected_OS}, using~\cite[Lemma~$3$]{IRS-OC} for large $K$, we can show that  $\max_{k\in [k]} \left|\gamma^{\mathsf{C}}_{k}\right|^2$ grows as $l_K$, where $F(l_K) = 1-\frac{1}{K}$ with $F(z) = 1-2\sqrt{z}K_1(2\sqrt{z})$ and $K_1(\cdot)$ is the first order modified Bessel function of the second kind. Since inverting $F(z)$ is difficult, we numerically approximate its inverse in the region of interest and obtain $l_K = (t\ln{K})^q$, where $t=0.7498$ and $q=1.71$. Then, from~\eqref{eq_init_E_val}, the throughput scales as
 \begin{equation}
\!\!\!\!R^{(K)} \triangleq \mathbb{E}[R_{\textrm{MR}}] \lesssim W\log_2\left(1+\frac{\rho G_{\mathsf{Tx}}G_{\mathsf{Rx}}P}{N\sigma^2}M^2 (t\ln{K})^q\right),
 \end{equation}
which is restated in the statement of the theorem.
    \end{proof}
 The above theorem confirms that a channel gain of $\mathcal{O}(M^2)$ can be obtained on all SCs by exploiting the B-SP effects at the IRS. Moreover, we get an additional boost by a factor of $\ln K$, which is the benefit of multi-user diversity, similar to that derived in~\cite[Theorem~$4$]{IRS-OC} for an opportunistic OFDMA in sub-$6$ GHz systems. 
 
 Note that Theorem~\ref{thm:Rate_scaling} is obtained using a large number of UEs to seek full array gain of $M^2$ on all SCs. However, a practically useful aspect to understand is the implication of the scheme for a finite number of UEs in the system. We have the following result. 
\begin{corollary}\label{thm:user_scaling}
    The number of users $K_{\textrm{min}}$ required in the system so that, on every SC, at least one UE obtains an array gain of $(1-\epsilon)M^2$ with probability at least $1-\delta$ must satisfy
    \begin{equation}\label{eqn_min_K}
      K   \geq K_{\textrm{min}} \triangleq -\ln\left(N/\delta\right)\Big/\ln{\Big(1-\frac{\sqrt{3\epsilon}}{\pi M(1+W/2f_c)}\bigg)}.
    \end{equation}
\end{corollary}
\begin{proof}
The proof follows by setting the probability derived in Theorem~\ref{thm_prob_success} to $1-\delta$ and rearranging the terms.
\end{proof}
From~\eqref{eqn_min_K}, $K_{\textrm{min}}$ increases with $\epsilon$ and $\delta$ decrease, in line with Theorems~\ref{thm_prob_success} and~\ref{thm:Rate_scaling} (see also~\cite[Prop.~$1$]{Yashvanth_ICASSP_2023}.) Also, $K_{\textrm{min}}$ increases with $M$ because the B-SP effect is more pronounced for larger $M$, and to exploit the B-SP we need more UEs. Since we desire full array gain on all SCs simultaneously, $K_{\textrm{min}}$ increases with $N$ as well. 
\begin{remark}
In our proposed scheme, to implement the opportunistic scheduler, the BS should know the index of the UE that obtains the best channel gain on each SC. This can be performed using an efficient low-complexity feedback schemes on a SC-by-SC basis. For e.g., the timer-feedback mechanism in~\cite{timer-scheme} uses a single pilot, and the UE with the best channel quality feeds back to the BS much earlier than other UEs. Therefore, our proposed scheme reaps optimal IRS benefits and obviates channel estimation and phase optimization algorithms, making it attractive for implementation.
\end{remark}

\section{Numerical Results}\label{sec:numerical-results}

\begin{figure*}
\vspace{-0.4cm}
\hspace{-0.7cm}
 \begin{subfigure}{0.33\linewidth}
 \centering
\includegraphics[width=1.13\linewidth]{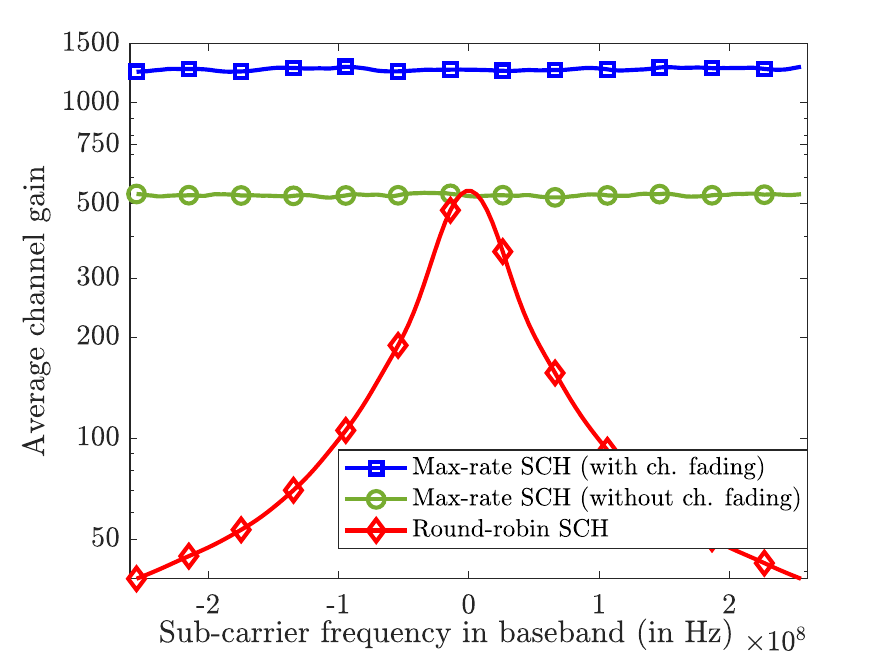}
     \caption{Avg. channel gain with $K=5000$, $M=512$.}
     \label{fig:avg_channel_1}
\end{subfigure}
\hspace{0.05cm}
\begin{subfigure}{0.33\linewidth}
\centering
\includegraphics[width=1.13\linewidth]{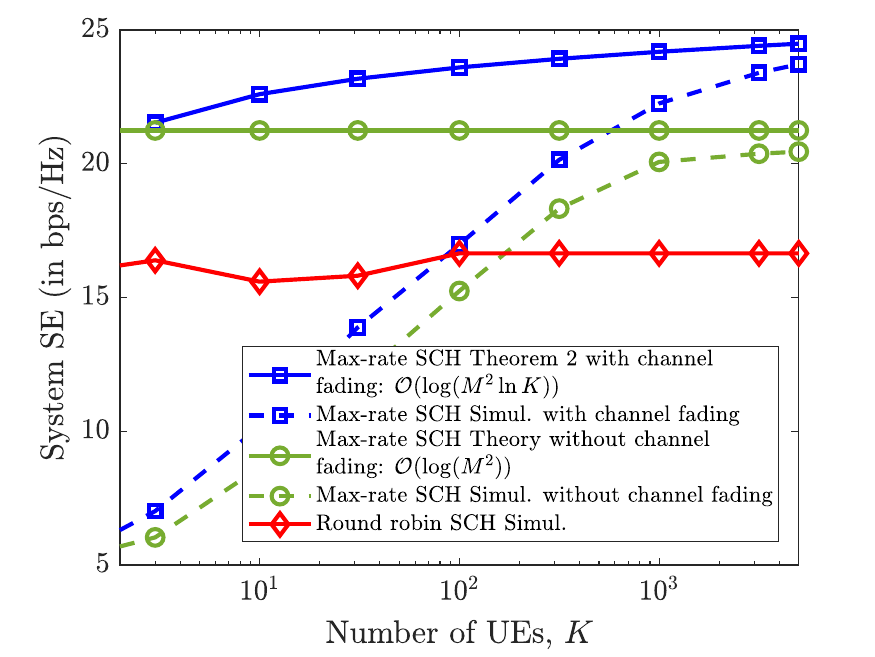}
     \caption{SE vs. No. of UEs, $K$ for $M=512$.}
\label{fig_ergodic_rate}
\end{subfigure}
\hspace{0.05cm}
\begin{subfigure}{0.33\linewidth}
\centering
\includegraphics[width=1.13\linewidth]{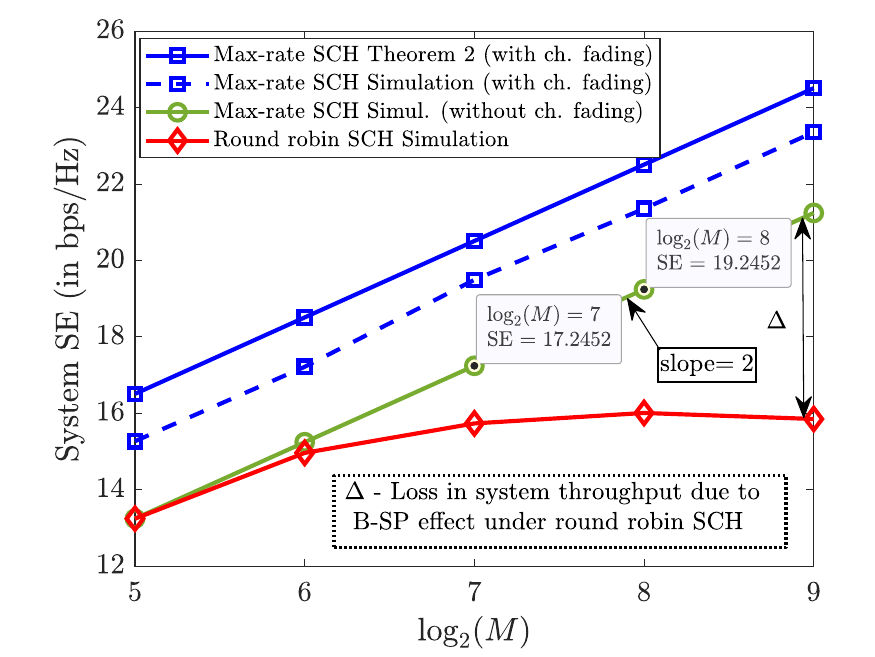}
     \caption{SE vs. IRS elements, $M$ for $K=5000$.}
\label{fig:rate_scaling_elements_1}
\end{subfigure}
\caption{Illustration of exploiting the B-SP effect to enhance the OFDMA performance with the number of SCs, $N=128$.}
\end{figure*}

We numerically illustrate the throughput enhancement obtained by exploiting the B-SP effect via OFDMA using Monte Carlo simulations. The BS is located at $(0,0)$ and an $M = 512$ element IRS is located at $(0, 500)$ (meters). The $K$ users are uniformly distributed around an annular region centered at the IRS with inner and outer radii $0.5$ km and $1$ km, respectively. The path loss exponent for the BS-IRS and IRS-UE links are $2$ and $4$, respectively. The antenna gain at BS and UE are $G_{\mathsf{Tx}} = 20$ dBi and $G_{\mathsf{Rx}} = 10$ dBi, respectively. The transmit power is $P = 40$ dBm, and the noise variance at UE is $\sigma^2 = -110$ dBm, with $N=128$ SCs spanning a bandwidth of $W = 510$ MHz at carrier frequency $f_c = 30$ GHz~\cite{beam-squint-1}.

In Fig.~\ref{fig:avg_channel_1}, we plot the average channel gain versus the baseband SC frequency for three scenarios: $1)$ a single UE is scheduled in every time slot over the total BW using round-robin (RR) scheduling and the IRS is optimized to the scheduled UE at $f_n=0$ in every slot, $2)$ the IRS is randomly tuned as per~\eqref{eqn:22} and the UEs are served over an OFDMA using a max-rate scheduler without considering channel fading effects (i.e., $\gamma_k^{\mathsf{C}}$ is deterministic), and $3)$ the IRS is randomly tuned as per~\eqref{eqn:22} and the UEs are served over an OFDMA using a max-rate scheduler including the random channel fading effects. The array gain using the RR scheduler peaks at $\mathcal{O}(M^2)$ only for $f_n=0$ (corresponding to a channel gain of around $500$ in this case), and it degrades on other SCs due to the B-SP effect. When multiple UEs are multiplexed on an OFDMA, the channel gain, as seen by the BS, flattens over the BW. In particular, when fading is absent, the gain flattens at the peak value itself, indicating the benefit of exploiting B-SP via OFDMA. With fading, the array gain further increases to $\mathcal{O}(M^2\ln (K))$ due to multi-user diversity and clearly surpasses the peak gain obtained with the RR scheduler.

Next, in Fig.~\ref{fig_ergodic_rate}, we evaluate the system spectral efficiency (SE), which is the achievable throughput per unit BW, as a function of the number of UEs for the three scenarios described in the previous paragraph, with $M=512$. The RR SE is smaller than the peak scaling of $\mathcal{O}(\log_2(M^2)$ due to the B-SP effect. Further, the SE does not vary with the number of UEs since RR scheduling does not exploit the multi-user diversity. However, with a max-rate scheduler, the SE increases with the number of UEs, $K$. When channel fading is absent, the SE eventually converges to the optimal rate scaling of $\mathcal{O}(\log_2(M^2))$ on all SCs for a large number of UEs, in line with Theorem~\ref{thm_prob_success}. When fading is present, multi-user diversity provides a further improvement to $\mathcal{O}(\log_2(M^2\ln K))$, in line with Theorem~\ref{thm:Rate_scaling}. Thus, using the B-SP effect with a large number of UEs, opportunistic scheduling of UEs on each SC procures optimal benefits on all SCs simultaneously. 

Finally, in Fig.~\ref{fig:rate_scaling_elements_1}, we illustrate the utility of the proposed scheme by plotting the system SE versus the number of IRS elements (in log scale) for the three scenarios. In all cases, the slope of the curves is equal to $2$ for smaller $M$, indicating that a full array gain of $M^2$ can be obtained on all SCs. However, as $M$ increases, the B-SP effect kicks in,  and the slope of the RR scheduler drops below $2$. Contrariwise, the slope of the curves corresponding to max-rate schedulers stays at $2$, affirming that the network achieves the full array gain from the IRS on all SCs by exploiting the B-SP effect. Also, the varying intercepts with and without fading effects are due to the additional benefits of multi-user diversity (of the order $\mathcal{O}(\log \ln K)$) in the former. 

\vspace{-0.2cm}
\section{Conclusions}\label{sec:conclusions}
In this paper, we demonstrated that the B-SP effect, which arises in IRS-aided wideband systems, can be exploited to enhance the system performance. We deduced that the B-SP effect causes the IRS to beamform at different angles on different SCs. Then, considering an OFDMA framework, we showed that, almost surely, we can obtain the optimal array gain of $M^2$ on all SCs when UEs are scheduled opportunistically using a max-rate scheduler, and there are sufficiently many UEs in the system. Subsequently, we derived the rate scaling laws of the system and proved that not only does our scheme achieve full array gain over the whole BW but also reaps the multi-user diversity gains in the system. Future work can include incorporating fairness in UE scheduling, addressing guaranteed rate requirements, and developing efficient signaling schemes to enable opportunistic scheduling of the UEs.

\bibliographystyle{IEEEtran}
\bibliography{IEEEabrv,IRS_ref_short}

\begin{thebibliography}{10}
\providecommand{\url}[1]{#1}
\csname url@samestyle\endcsname
\providecommand{\newblock}{\relax}
\providecommand{\bibinfo}[2]{#2}
\providecommand{\BIBentrySTDinterwordspacing}{\spaceskip=0pt\relax}
\providecommand{\BIBentryALTinterwordstretchfactor}{4}
\providecommand{\BIBentryALTinterwordspacing}{\spaceskip=\fontdimen2\font plus
\BIBentryALTinterwordstretchfactor\fontdimen3\font minus
  \fontdimen4\font\relax}
\providecommand{\BIBforeignlanguage}[2]{{%
\expandafter\ifx\csname l@#1\endcsname\relax
\typeout{** WARNING: IEEEtran.bst: No hyphenation pattern has been}%
\typeout{** loaded for the language `#1'. Using the pattern for}%
\typeout{** the default language instead.}%
\else
\language=\csname l@#1\endcsname
\fi
#2}}
\providecommand{\BIBdecl}{\relax}
\BIBdecl

\bibitem{IRS_phase_shift}
Q.~Wu and R.~Zhang, ``Towards smart and reconfigurable environment: Intelligent
  reflecting surface aided wireless network,'' \emph{{IEEE} Commun. Mag.},
  vol.~58, no.~1, pp. 106--112, Jan. 2020.

\bibitem{narrowband-2}
J.~Lin, G.~Wang, R.~Fan, T.~A. Tsiftsis, and C.~Tellambura, ``Channel
  estimation for wireless communication systems assisted by large intelligent
  surfaces,'' \emph{arXiv preprint arXiv:1911.02158}, 2019.

\bibitem{narrowband-3}
Z.~Wang, L.~Liu, and S.~Cui, ``Channel estimation for intelligent reflecting
  surface assisted multiuser communications: Framework, algorithms, and
  analysis,'' \emph{{IEEE} Trans. Wireless Commun.}, vol.~19, no.~10, pp.
  6607--6620, Oct. 2020.

\bibitem{Yashuai_PIMRC_2020}
Y.~Cao, T.~Lv, and W.~Ni, ``Intelligent reflecting surface aided multi-user
  mmwave communications for coverage enhancement,'' in \emph{Proc. IEEE 31st
  Annu. Int. Symp. Pers., Indoor Mobile Radio Commun.}, Aug. 2020, pp. 1--6.

\bibitem{VanTress_Array_Proc_2002}
H.~L. Van~Trees, \emph{Optimum array processing: Part {IV} of detection,
  estimation, and modulation theory}.\hskip 1em plus 0.5em minus 0.4em\relax
  John Wiley \& Sons, 2002.

\bibitem{beam-squint-1}
S.~Ma, W.~Shen, J.~An, and L.~Hanzo, ``Wideband channel estimation for
  {IRS}-aided systems in the face of beam squint,'' \emph{{IEEE} Trans.
  Wireless Commun.}, vol.~20, no.~10, pp. 6240--6253, Oct. 2021.

\bibitem{beam-squint-2}
Y.~Chen, D.~Chen, and T.~Jiang, ``Beam-squint mitigating in reconfigurable
  intelligent surface aided wideband mmwave communications,'' in \emph{Proc.
  IEEE Wireless Commun. Netw. Conf. (WCNC)}, 2021, pp. 1--6.

\bibitem{beam-squint-3}
L.~Afeef and H.~Arslan, ``Beam squint effect in multi-beam mmwave massive
  {MIMO} systems,'' in \emph{Proc. IEEE 96th Veh. Technol. Conf. (VTC-Fall)},
  2022, pp. 1--5.

\bibitem{Feifei_Dual_wideband_TSP_2018}
B.~Wang, F.~Gao, S.~Jin, H.~Lin, and G.~Y. Li, ``Spatial- and
  frequency-wideband effects in millimeter-wave massive {MIMO} systems,''
  \emph{{IEEE} Trans. Signal Process.}, vol.~66, no.~13, pp. 3393--3406, Jul.
  2018.

\bibitem{Jiang_WCNC_2021_BSQ}
Y.~Chen, D.~Chen, and T.~Jiang, ``Beam-squint mitigating in reconfigurable
  intelligent surface aided wideband mmwave communications,'' in \emph{Proc.
  IEEE Wireless Commun. Netw. Conf. (WCNC)}, Mar. 2021, pp. 1--6.

\bibitem{Derrick_Wing_Kwan_TCOM_2023}
R.~Su, L.~Dai, and D.~W.~K. Ng, ``Wideband precoding for {RIS}-aided {THz}
  communications,'' \emph{{IEEE} Trans. Commun.}, vol.~71, no.~6, pp.
  3592--3604, Jun. 2023.

\bibitem{Wanming_TVT_2023}
W.~Hao, X.~You, F.~Zhou, Z.~Chu, G.~Sun, and P.~Xiao, ``The far-/near-field
  beam squint and solutions for {THz} intelligent reflecting surface
  communications,'' \emph{{IEEE} Trans. Veh. Technol.}, vol.~72, no.~8, pp.
  10\,107--10\,118, Aug. 2023.

\bibitem{Zhao_WCL_2024}
F.~Zhao, W.~Hao, X.~You, Y.~Wang, Z.~Chu, and P.~Xiao, ``Joint beamforming
  optimization for {IRS}-aided {THz} communication with time delays,''
  \emph{{IEEE} Wireless Commun. Lett.}, vol.~13, no.~1, pp. 49--53, Jan. 2024.

\bibitem{Octavia_CL_2022}
H.~Sun, S.~Zhang, J.~Ma, and O.~A. Dobre, ``Time-delay unit based beam squint
  mitigation for {RIS}-aided communications,'' \emph{{IEEE} Commun. Lett.},
  vol.~26, no.~9, pp. 2220--2224, Sep. 2022.

\bibitem{Beixiong_TWC_2020_OFDMA}
B.~Zheng, C.~You, and R.~Zhang, ``Intelligent reflecting surface assisted
  multi-user {OFDMA}: Channel estimation and training design,'' \emph{{IEEE}
  Trans. Wireless Commun.}, vol.~19, no.~12, pp. 8315--8329, Dec. 2020.

\bibitem{Yang_WCL_2020_OFDMA}
Y.~Yang, S.~Zhang, and R.~Zhang, ``{IRS}-enhanced {OFDMA}: Joint resource
  allocation and passive beamforming optimization,'' \emph{{IEEE} Wireless
  Commun. Lett.}, vol.~9, no.~6, pp. 760--764, Jun. 2020.

\bibitem{pathloss}
Z.~Wan, Z.~Gao, and M.-S. Alouini, ``Broadband channel estimation for
  intelligent reflecting surface aided mmwave massive {MIMO} systems,'' in
  \emph{Proc. IEEE Int. Conf. Commun. (ICC)}, Jun. 2020, pp. 1--6.

\bibitem{Yashvanth_TCOM_2024}
\BIBentryALTinterwordspacing
L.~Yashvanth and C.~R. Murthy, ``On the impact of an {IRS} on the out-of-band
  performance in sub-6 {GHz} \& mmwave frequencies,'' \emph{{IEEE} Trans.
  Commun.}, pp. 1--1, 2024. [Online]. Available:
  \url{https://ieeexplore.ieee.org/document/10552892}
\BIBentrySTDinterwordspacing

\bibitem{Yashvanth_SPCOM_2022}
L.~Yashvanth, C.~R. Murthy, and D.~Battu, ``Binary intelligent reflecting
  surfaces assisted {OFDM} systems,'' in \emph{Proc. IEEE Int. Conf. Signal
  Process. Commun. (SPCOM)}, Jul. 2022, pp. 1--5.

\bibitem{IRS-OC}
L.~Yashvanth and C.~R. Murthy, ``Performance analysis of intelligent reflecting
  surface assisted opportunistic communications,'' \emph{{IEEE} Trans. Signal
  Process.}, vol.~71, pp. 2056--2070, 2023.

\bibitem{Nadeem_TWC_2021}
Q.-U.-A. Nadeem, A.~Zappone, and A.~Chaaban, ``Intelligent reflecting surface
  enabled random rotations scheme for the {MISO} broadcast channel,''
  \emph{{IEEE} Trans. Wireless Commun.}, vol.~20, no.~8, pp. 5226--5242, Aug.
  2021.

\bibitem{Yashvanth_ICASSP_2023}
L.~Yashvanth and C.~R. Murthy, ``Comparative study of {IRS} assisted
  opportunistic communications over i.i.d. and {LoS} channels,'' in \emph{Proc.
  IEEE Int. Conf. Acoust., Speech Signal Process. (ICASSP)}, Jun. 2023, pp.
  1--5.

\bibitem{timer-scheme}
V.~Shah, N.~B. Mehta, and R.~Yim, ``Optimal timer based selection schemes,''
  \emph{{IEEE} Trans. Commun.}, vol.~58, no.~6, pp. 1814--1823, Jun. 2010.

\end{thebibliography}
\end{document}